\documentclass[11pt,letterpaper]{article}

\usepackage[margin=1.15in]{geometry}

\usepackage[utf8]{inputenc}
\usepackage[T1]{fontenc}
\usepackage{lmodern}

\usepackage{amsmath,amssymb, amsfonts, amsthm, mathtools, bm} 
\usepackage{empheq}
\usepackage{thm-restate}
\newtheorem{theorem}{Theorem}

\newtheorem{lemma}[theorem]{Lemma}

\usepackage{calc}
\usepackage{tikz}
\usetikzlibrary{calc}
\usetikzlibrary{decorations.markings}
\usetikzlibrary{arrows,calc,decorations.pathmorphing}
\usetikzlibrary{decorations.pathreplacing,shapes.misc}
\usetikzlibrary{positioning}

\def\Lscr{\mathcal{L}}
\def\Nscr{\mathcal{N}}
\def\Wscr{\mathcal{W}}
\def\Sscr{\mathcal{S}}

\newcommand{\symdiff}{\bigtriangleup}
\renewcommand{\epsilon}{\varepsilon}
\newcommand{\odd}{\ensuremath{\mathrm{odd}}}
\def\cupp{\stackrel{.}{\cup}}

\usepackage{xcolor}

\usepackage[vlined,ruled,algo2e]{algorithm2e}
\SetAlCapHSkip{1em}
\title{Improving on Best-of-Many-Christofides for $T$-tours}
\author{Vera Traub\thanks{
Department of Mathematics, ETH Zurich, Zurich, Switzerland.
Email: {vera.traub@ifor.math.ethz.ch}.
Supported by Swiss National Science Foundation grant 200021\_184622.
}}
\date{}

\begin{document}

 \maketitle

\begin{abstract}
 The $T$-tour problem is a natural generalization of TSP and Path TSP.
 Given a graph $G=(V,E)$, edge cost $c: E \to \mathbb{R}_{\ge 0}$, and an even cardinality set $T\subseteq V$,
 we want to compute a minimum-cost $T$-join connecting all vertices of $G$ (and possibly containing parallel edges). 
 
 In this paper we give an $\frac{11}{7}$-approximation for the $T$-tour problem and show that the integrality ratio of the standard LP relaxation  is at most $\frac{11}{7}$.
 Despite much progress for the special case Path TSP, for general $T$-tours this is the first improvement 
 on Seb\H{o}'s analysis of the Best-of-Many-Christofides algorithm (Seb\H{o}~[2013]).
\end{abstract}

\section{Introduction}

The traveling salesman problem (TSP) is one of the most classical problems in combinatorial optimization.
Given a set $V$ of vertices and a metric $c$ on $V$, we want to find an order $v_1,\dots, v_n$ of the vertices in $V$ minimizing 
$c(v_n,v_1) +\sum_{i=2}^n c(v_{i-1},v_i)$.
Another definition of the TSP, which can be easily shown to be equivalent, is the following.
Given a connected graph $G=(V,E)$ and nonnegative edge costs $c : E \to \mathbb{R}_{\ge 0}$, find a minimum cost multi-subset $F$ of $E$ such that 
$(V,F)$ is connected and Eulerian, i.e. every vertex has even degree.
For many years best known approximation algorithm for the TSP was the classical $\frac{3}{2}$-approximation algorithm due to Christofides~\cite{Chr76} and Serdjukov~\cite{Ser78}.
Only very recently this approximation ratio was improved to $\frac{3}{2}-\epsilon$ for some small $\epsilon > 0$ by Karlin, Klein, and Oveis Gharan~\cite{KarKOG20}.

One important variant of the TSP is the Path TSP.
Besides a set $V$ of vertices and a metric~$c$ on $V$, we are given a start-vertex $s\in V$ and an end-vertex $t\in V$.
The task is to find an order $s=v_1,\dots, v_n=t$ of the vertices in $V$ minimizing $\sum_{i=2}^n c(v_{i-1},v_i)$.
As for the TSP we can also formulate the Path TSP as a graph problem:
given a connected graph $G=(V,E)$, vertices $s,t\in V$ ($s\ne t$), and nonnegative edge costs $c : E \to \mathbb{R}_{\ge 0}$, find a minimum cost multi-subset $F$ of $E$ such that 
$(V,F)$ is connected and the set $\odd(F)$ of odd-degree vertices in $(V,F)$ contains precisely the vertices $s$ and $t$.
In other words, $s$ and $t$ have odd degree, while all other vertices have even degree.

Christofides' algorithm for the TSP can be generalized to the path version, but then it has an approximation ratio of only $\frac{5}{3}$
as shown by Hoogeveen~\cite{Hoo91}.
However in contrast to TSP, for the path version we do know better approximation algorithms than Christofides' algorithm.
The first such approximation algorithm was given by An, Kleinberg, and Shmoys~\cite{AnKS15}, who proposed and analyzed the 
Best-of-Many-Christofides algorithm, which we will discuss in more detail later in this paper.
Subsequently, there has been a line of work~\cite{Seb13,Vyg16,GotV16,SebvZ16,TraV18,Zen18} improving the approximation ratio further.
%At the moment the best-known approximation ratio is $\frac{3}{2}$ due to Zenklusen~\cite{Zen18}.
Moreover, there is a black-box reduction from the path version to TSP~\cite{TraVZ20}:
if there is an $\alpha$-approximation algorithm for TSP, there also is an $(\alpha +\epsilon)$-approximation algorithm for the path version, for any fixed $\epsilon > 0$.
Combining this black-box reduction with their new approximation algorithm for TSP, Karlin, Klein, and Oveis Gharan~\cite{KarKOG20} obtain a $(\frac{3}{2}-\epsilon)$-approximation algorithm for Path TSP for some small $\epsilon > 0$. This is the currently best-known approximation ratio for Path TSP.

In this paper we study the $T$-tour problem which is a natural generalization of TSP and its path version.
An instance consists of a connected graph $G=(V,E)$, a set $T\subseteq V$ with $|T|$ even, and nonnegative edge costs $c : E \to \mathbb{R}_{\ge 0}$. 
The task is to compute a a minimum cost multi-subset $F$ of $E$ such that $(V,F)$ is connected and $\odd(F)=T$, i.e.\ vertices in $T$ have odd degree and vertices in $V\setminus T$ have even degree.
In other words, $F$ is a $T$-join (with possibly parallel edges) and connects all vertices of $G$.
The TSP is the special case $T=\emptyset$ and the Path TSP is the special case $|T|=2$.

However, many of the results for the Path TSP do not generalize to the $T$-tour problem.
Cheriyan, Friggstad, and Gao \cite{CheFG15} extended the Best-of-Many-Christofides algorithm by An, Kleinberg, and Shmoys~\cite{AnKS15}
from $|T|= 2$ to general $T$ and proved an approximation ratio of $\frac{13}{8}$, which is slightly worse than the ratio $\frac{1+\sqrt{5}}{2}$ obtained in~\cite{AnKS15} for $|T|= 2$.
Then Seb\H{o} \cite{Seb13} improved the analysis of the same algorithm and showed that the Best-of-Many-Christofides algorithm 
yields an $\frac{8}{5}$-approximation for the $T$-tour problem, which was also an improvement for the Path TSP.
Despite much further progress regarding the approximability of Path TSP, 
this result is the best previously known approximation ratio for the $T$-tour problem.

The results from \cite{GotV16,SebvZ16,TraV19,Zho20} apply only to Path TSP, i.e.\ the case $|T|=2$.
The reason for this is that they all rely on a structural theorem by Gottschalk and Vygen~\cite{GotV16}, which
cannot be extended for the case $|T|\ge 4$ as also shown in~\cite{GotV16}.

Some of the results for the Path TSP that are based on a dynamic programming technique~\cite{TraV18,Zen18,TraVZ20}
can be extended to the $T$-tour problem with $|T|$ constant, but not to the general case.
This yields a $(\frac{3}{2} - \epsilon)$-approximation algorithm for some small $\epsilon > 0$ and $|T|$ constant.
For unit-weight graphs, i.e. the special case where $c(e) =1$ for every edge $e\in E$, 
a $\frac{3}{2}$-approximation algorithm is known for general $|T|$ \cite{SebV14}.

In this paper we give the first improvement of the approximation guarantee for the general $T$-tour problem over Seb\H{o}'s 
$\frac{8}{5}$-approximation algorithm \cite{Seb13}. 
Our main result is an $\frac{11}{7}$-approximation algorithm.
We analyze the algorithm with respect to the standard LP relaxation (see~\eqref{eq:t_tour_lp} in Section~\ref{sect:preliminaries})
and we therefore also prove an upper bound of $\frac{11}{7}$ on the integrality ratio of this relaxation.

\section{Best-of-Many-Christofides and lonely edge deletion}\label{sect:preliminaries}
We will analyze our algorithm with respect to the following LP relaxation.
  \begin{equation}\label{eq:t_tour_lp}
  \begin{aligned}
   \min c(x)  \\
   s.t. & & x(\delta(U))\ge&\ 2 & &\text{ for }\emptyset\neq U\subsetneq V\text{ with }|T \cap U|\text{ even} \\
   & &  x(\delta(\Wscr)) \ge&\ |\Wscr| - 1 & & \text{ for every partition }\Wscr\text{ of }V \\
   & & x_e \ge & \ 0 & &\text{ for } e\in E,
  \end{aligned}
 \end{equation}
 where $\delta(U)$ is the set of edges with exactly one endpoint in $U$ and 
 $\delta(\Wscr)$ denotes the set of edges with endpoints in different elements of the partition $\Wscr$.
 The LP \eqref{eq:t_tour_lp} can be solved in polynomial time using the ellipsoid method.
 Barahona and Conforti \cite{BarC87} show that one can separate the even cut constraints 
 $x(\delta(U))\ge 2$ for $\emptyset\neq U\subsetneq V$ with $|T\cap U|$ even in polynomial time.
 The other constraints define the connector polyhedron
 \begin{equation}\label{eq:conncetor_polyhedron}
  \left\{ x \in \mathbb{R}^E_{\ge 0} : x(\delta(\Wscr)) \ge |\Wscr| - 1 \text{ for every partition }\Wscr\text{ of  }V\right\},
 \end{equation}
 which is the convex hull of all multi-subsets $F$ of $E$ for which $(V,F)$ is connected. (See e.g.\ Section~50.5 in~\cite{Sch03}.)
 Since one can optimize over \eqref{eq:conncetor_polyhedron} in polynomial time, one can also separate its constraints in polynomial time.
 
 Let $x^*$ be an optimum solution to~\eqref{eq:t_tour_lp}.
 In the following we denote by $\Sscr$ the set of all edge sets of spanning trees of our given graph $G$. 
 Since the constraints of~\eqref{eq:t_tour_lp} imply that $x^*$ is contained in the connector polyhedron \eqref{eq:conncetor_polyhedron},
 the vector $x^*$ dominates a convex combination of incidence vectors of spanning trees, i.e.\ there are coefficients
 $p_S \ge 0$ for $S\in \Sscr$ such that $\sum_{S\in \Sscr} p_S = 1$ and 
 \begin{equation}\label{eq:lp_sol_dominates_convex_comb}
  x^* \ge \sum_{S\in \Sscr} p_S \cdot \chi^S.
 \end{equation}

To prove our main result, we will use two different algorithms and bound the cost of the better of the two resulting $T$-tours.
One of these two algorithms is the Best-of-Many-Christofides algorithm, proposed by An, Kleinberg, and Shmoys~\cite{AnKS15} for the Path TSP and extended to $T$-tours by Cheriyan, Friggstad, and Gao~\cite{CheFG15}.
The algorithm proceeds as follows.

\medskip
\begin{algorithm2e}[H]
\vspace*{1mm}
\begin{enumerate}\itemsep0pt
\item Compute an optimum solution $x^*$ to the LP~\eqref{eq:t_tour_lp}.
\item Find a convex combination \eqref{eq:lp_sol_dominates_convex_comb} of spanning trees dominated by $x^*$.
\item For every $S\in \Sscr$ with $p_S > 0$, compute a cheapest $(\odd(S)\symdiff T)$-join $J_S^*$.
\item Return the cheapest of the resulting $T$-tours $S\cupp J_S^*$.
\end{enumerate}
\vspace*{-2mm}
\caption{Best-of-Many-Christofides \label{algo:bomc}
}
\end{algorithm2e}
\medskip

The cost of the tour computed by Algorithm~\ref{algo:bomc} is
\begin{align*}
 \min_{S\in \Sscr: p_S >0} \left(c(S) + c(J_S^*)\right) \le\ & \sum_{S \in \Sscr} p_S \cdot \left(c(S) + c(J_S^*)\right) 
                                                      \le c(x^*) + \sum_{S \in \Sscr} p_S \cdot c(J_S^*).
\end{align*}
To bound the cost of the $(\odd(S)\symdiff T)$-join $J_S^*$, both \cite{AnKS15} and \cite{Seb13} follow Wolsey's analysis~\cite{Wol80} of 
Christofides' algorithm for TSP. 
Since every vector contained in the $(\odd(S)\symdiff T)$-join polyhedron
\begin{equation}\label{eq:T_join_polyhedron}
  \Big\{ y\in \mathbb{R}_{\ge 0}^E : y(\delta(U)) \ge 1 \text{ for every } U\text{ with } |U\cap(\odd(S)\symdiff T)|\text{ odd} \Big\}
\end{equation}
dominates a convex combination of incidence vectors of $(\odd(S)\symdiff T)$-joins
\cite{EdmJ73}, we have $c(J_S^*) \le c(y^S)$ for every vector $y^S$ in \eqref{eq:T_join_polyhedron}.
We call a vector in \eqref{eq:T_join_polyhedron} a \emph{parity correction vector}.
The main difficulty in the analysis of the Best-of-Many-Christofides algorithm is to construct a cheap parity correction vector.

If $T=\emptyset$ (which is the special case TSP), the vector $\frac{1}{2}x^*$ is a feasible parity correction vector.
To see this, note that in this case $|T\cap U|$ is even for every $\emptyset\neq U\subsetneq V$ and hence $x^*(C)\ge 2$
for every cut $C$.
However, for $T\ne \emptyset$ this is not necessarily the case.
We call a cut $C$ with $x^*(C) < 2$ \emph{narrow} and denote by
\[ \Nscr := \{ \delta(U) : x^*(\delta(U)) < 2 \} \]
the set of \emph{narrow cuts}.
By the constraints of~\eqref{eq:t_tour_lp}, every narrow cut $C$ is a $T$-cut, i.e.\ $C=\delta(U)$ for some $U\subset V$ with $|U\cap T|$ odd.

Recall that for a parity correction vector $y$ we require $y(C)\ge 1$ only if $C$ is an $(\odd(S)\symdiff T)$-cut.
Let $U\subseteq V$ such that $\delta(U)\in \Nscr$ is a narrow cut. 
Then $|T\cap U|$ is odd.
Hence, $|U \cap (\odd(S)\symdiff T)|$ is odd if and only if $|U\cap \odd(S)|$ is even.
This is the case if and only if $|\delta(U) \cap S|$ is even.
In particular, a narrow cut $C$ with $|S\cap C|=2$ is an $(\odd(S)\symdiff T)$-cut, while a narrow cut $C$ with $|S\cap C|=1$ is not.
% (In the worst case of Seb\H{o}'s \cite{Seb13} analysis of Algorithm~\ref{algo:bomc} we have $|S\cap C| \in \{1,2\}$ for all almost all narrow cuts.)

If $|S \cap C| = 1$ for a narrow cut $C$, we say that $C$ is \emph{lonely} for $S$.
Then we also say that the unique edge $e \in C\cap S$ is \emph{lonely} at $C$.
Lonely cuts and edges play a special role in Seb\H{o}'s \cite{Seb13} analysis of the Best-of-Many-Christofides algorithm in two ways.
First, the lonely cuts of a tree $(V,S)$ are important since they don't need parity correction, meaning that they are no $(\odd(S)\symdiff T)$-cuts.
Second, the incidence vectors $\chi^e$ of lonely edges are used to construct cheap parity correction vectors.
Here, the vector $\chi^e$ for an edge $e$ that is lonely in a tree $(V,S)$ is used to construct the parity correction vectors $y^{S'}$ for  other trees $(V,S')$. 
(See Section~\ref{sect:bomc} for more details.)

Besides the Best-of-Many-Christofides algorithm we will analyze another algorithm for the $T$-tour problem to prove our main result.
This algorithm builds on an algorithm by Seb\H{o} and van~Zuylen~\cite{SebvZ16} for the case $|T|=2$.
They start with a particular convex combination of incidence vectors of spanning trees, which was shown to exist by Gottschalk and Vygen~\cite{GotV16} 
(see also \cite{SchXX18}).
For every spanning tree $S$ contributing to the convex combination, they now delete some edges to obtain a forest $F_S$.
Then they compute an $(\odd(F_S) \symdiff T)$-join $J_S^*$.
Now $\odd(F_S\cupp J_S^*) =T$, but $(V, F_S\cupp J_S^*)$ might be disconnected. 
Therefore, they finally reconnect by adding two copies of some edges to obtain a $T$-tour. 

Seb\H{o} and van Zuylen~\cite{SebvZ16} show that on average they save more by deleting edges than they need to pay in the final reconnection step.
The reason why their result does not carry over to general $T$-tours is that their analysis crucially relies on the structure of the convex combination of incidence vectors of spanning trees which one cannot achieve for $|T| \ge 4$ (see \cite{GotV16}).    
     
\section{Outline of our approach}\label{sect:our_contribution}

Let us now explain our new approximation algorithm for the $T$-tour problem and outline its analysis.
We proceed as in the algorithm by Seb\H{o} and van Zuylen~\cite{SebvZ16} for the case $|T|=2$, but we start with an 
arbitrary convex combination \eqref{eq:lp_sol_dominates_convex_comb} instead of one with additional structure.
In the following we denote by $\Lscr_S$ the set of lonely cuts of a spanning tree $(V,S)$, 
i.e.\ the set of narrow cuts $C\in \Nscr$ with $|S\cap C|=1$.
Moreover, $L_S$ denotes the set of lonely edges.

\medskip
\begin{algorithm2e}[H]
\vspace*{1mm}
\begin{enumerate}\itemsep0pt
\item Compute an optimum solution $x^*$ to the LP~\eqref{eq:t_tour_lp}.
\item Find a convex combination \eqref{eq:lp_sol_dominates_convex_comb} of spanning trees dominated by $x^*$
      and \\ compute the set $\Nscr$ of narrow cuts. \label{item:narrow_cut_computation}
\item For every $S\in \Sscr$ with $p_S > 0$, let $L_S = \bigcup_{C\in \Lscr_S} (S\cap C)$ and 
      $F_S := S\setminus L_S$. \\
      Compute an $(\odd(F_S)\symdiff T)$-join $J_S^*$ with minimum $c^S(J_S^*)$, where for an edge $e$
      \[
       c^S(e) := c(e) +  2 \cdot 
                       \left(\sum_{C\in\Lscr_S : e\in C} c(S\cap C) - \max_{C\in\Lscr_S : e\in C} c(S\cap C) \right),
      \]
      where $\max \emptyset := 0$.
\item Compute a cheapest set $R_S$ of edges such that $(V,F_S \cup J_S^* \cup R_S)$ is connected.
\item Return the cheapest of the resulting $T$-tours $S\cupp J_S^* \cupp R_S \cupp R_S$.
\end{enumerate}
\vspace*{-2mm}
\caption{Best-of-Many-Christofides with lonely edge deletion \label{algo:deleting_edges}
}
\end{algorithm2e}
In order to implement step~\ref{item:narrow_cut_computation} in polynomial time one can use any polynomial-time algorithm for enumerating all near-minimum cuts~\cite{Kar00,KarS96,NagNI97,VazY92}.

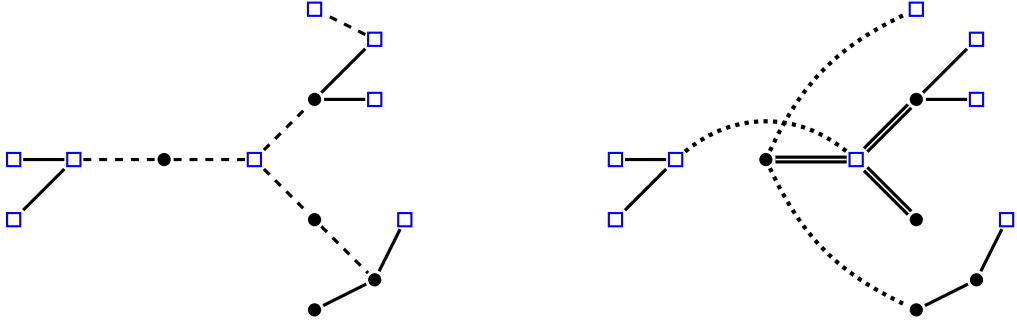
\begin{figure}
\begin{center}
\begin{tikzpicture}[scale=0.8]
\tikzset{nsT/.style={
draw=blue, thick,rectangle,inner sep=0em,minimum size=5pt,fill=white, outer sep=0.1em
}}
\tikzset{ns/.style={
fill=black,circle,inner sep=0em,minimum size=5pt, outer sep=0.1em}
} 
\tikzstyle{two}=[double,very thick]

\begin{scope}[every node/.style=nsT, shift={(9.5,5.5)}]
\node (v1) at (-1,0) {};
\node (v2) at (-1,-1) {};
\node (v3) at (0,0) {};
\node (v5) at (3,0) {};
\node (v7) at (5,1) {};
\node (v8) at (5,2) {};
\node (v9) at (4,2.5) {};
\node (v11) at (5.5,-1) {};
\end{scope}

\begin{scope}[every node/.style=ns, shift={(9.5,5.5)}]
\node (v4) at (1.5,0) {};
\node (v6) at (4,1) {};
\node (v10) at (4,-1) {};
\node (v12) at (5,-2) {};
\node (v13) at (4,-2.5) {};
\end{scope}

\begin{scope}[black,very thick]
\draw (v1) --(v3) --(v2);
\draw (v7) -- (v6) --(v8);
\draw (v11) --(v12) --(v13);
\end{scope}

\begin{scope}[black, dashed,very thick]
\draw (v3) -- (v4);
\draw (v4) -- (v5);
\draw (v5) -- (v6);
\draw (v5) -- (v10);
\draw (v8) -- (v9);
\draw (v10) -- (v12);
\end{scope}

\begin{scope}[shift={(10,0)}]
\begin{scope}[every node/.style=nsT, shift={(9.5,5.5)}]
\node (v1) at (-1,0) {};
\node (v2) at (-1,-1) {};
\node (v3) at (0,0) {};
\node (v5) at (3,0) {};
\node (v7) at (5,1) {};
\node (v8) at (5,2) {};
\node (v9) at (4,2.5) {};
\node (v11) at (5.5,-1) {};
\end{scope}

\begin{scope}[every node/.style=ns, shift={(9.5,5.5)}]
\node (v4) at (1.5,0) {};
\node (v6) at (4,1) {};
\node (v10) at (4,-1) {};
\node (v12) at (5,-2) {};
\node (v13) at (4,-2.5) {};
\end{scope}

\begin{scope}[black,very thick]
\draw (v1) --(v3) --(v2);
\draw (v7) -- (v6) --(v8);
\draw (v11) --(v12) --(v13);
\end{scope}

\begin{scope}
%\draw (v3) -- (v4);
\draw[two] (v4) -- (v5);
\draw[two] (v5) -- (v6);
\draw[two] (v5) -- (v10);
%\draw (v8) -- (v9);
%\draw (v10) -- (v11);
\end{scope}

\begin{scope}[black, dotted, ultra thick]
\draw[bend right=20] (v4) to (v13);
\draw[bend left=20] (v4) to (v9);
\draw[bend left=40] (v3) to (v5);
\end{scope}

\end{scope}

\end{tikzpicture}
\end{center}
\caption{The left picture shows a spanning tree $(V,S)$, where the edges in $F_S$ are solid and the lonely edges are dashed.
The squares with white interior represent the set $T$ and the filled circles represent $V\setminus T$.
The right picture shows a $T$-tour resulting from this tree.
It consists of the forest $F_S$, an $(\odd(F_S)\symdiff T)$-join $J_S$ (dotted), and $2 R_S$ (the doubled edges).
For every edge $\{v,w\}\in J_S$, the set $R_S$ contains two copies of all but one of the lonely edges of the $v$-$w$-path in $(V,S)$.
\label{fig:illustration_of_reconnect}}
\end{figure}

The cost function $c^S$ is chosen to anticipate the cost for reconnection.
More precisely, we can observe the following, which is shown in~\cite{SebvZ16} for $|T|=2$. See Figure~\ref{fig:illustration_of_reconnect}.
\begin{lemma}
 In Algorithm~\ref{algo:deleting_edges} we have for every $S\in \Sscr$ with $p_S > 0$
 \[ 
     c(J_S^*) + 2 \cdot c(R_S) \le c^S(J_S^*). 
 \]
\end{lemma}
\begin{proof}
 Since $J_S^*$ is an $(\odd(F_S) \symdiff T)$-join and every lonely cut of $S$ is an $(\odd(F_S) \symdiff T)$-cut,
 we have $|C\cap J_S^*| \ge 1$ for all $C\in \Lscr_S$.
 For an edge $e\in J_S^*$ let $R_e := \bigcup_{C \in \Lscr_S: e\in \Lscr_S} (C\cap S)$.
 Then $S \subseteq F_S \cup \bigcup_{e\in J_S^*} R_e$.
 
 For $e\in J_S^*$ let $R'_e$ result from $R_e$ by removing its most expensive element.
 Consider a lonely edge $l$ in the unique cycle in $S \cup \{e\}$.
 Then there is a unique (lonely) cut $C$ with $\{l\} = C \cap S$.
 For this cut $C$ we have $e\in C$, implying $l\in R_e$.
 Therefore, $F_S \cup R_e \cup \{e\}$ contains a cycle. 
 
 This implies that $F_S \cup J_S^* \cup \bigcup_{e\in J_S^*} R'_e$ is connected.
 Hence,
 \[
    c(J_S^*) + 2 \cdot c(R_S) \le c(J_S^*) + 2 \cdot \sum_{e\in J_S^*} c(R'_e) = c^S(J_S^*). \qedhere
 \]
\end{proof}

As in \cite{SebvZ16} we construct a vector $\bar y^S$ in the $(\odd(F_S)\symdiff T)$-join polyhedron
to bound the cost $c^S(J_S^*)$ of parity correction and reconnection.
However, we now have the following difficulty.
Like Seb\H{o}~\cite{Seb13} and Gottschalk and Vygen~\cite{GotV16},
 Seb\H{o} and van Zuylen~\cite{SebvZ16} also use the incidence vector $\chi^e$ of a lonely edge of tree $S$ to construct 
the parity correction vectors $\bar y^{S'}$ for other trees $S'$.
To bound the resulting expected reconnection cost $c^{S'}(e) - c(e)$, Seb\H{o} and van Zuylen exploit the particular structure of the convex combination they work with.
Without this structure we cannot give a sufficiently good bound anymore and hence we will not use 
the incidence vectors of lonely edges to construct parity correction vectors for other trees.

Instead, our parity correction vector $\bar y^S$ will consist only of a fraction of $x^*$ and incidence vectors of edges of the tree $S$ itself.
This will allow us to control the reconnection cost: to bound the cost $c^S(x^*)- c(x^*)$ we can generalize an argument from~\cite{SebvZ16} (see Lemma~\ref{lemma:reconnection_cost}) and for an edge $e\in S$ we have $c^S(e)=c(e)$.
%, because $e$ is not bad for $S$.
Since using a parity correction vector consisting not only of a fraction of $x^*$ and incident vectors of edges of the tree $S$ itself
was crucial in all of the improvements~\cite{AnKS15,Seb13,Vyg16,GotV16,SebvZ16} over Christofides' algorithm for $|T|=2$,
we need new insights to obtain a good bound anyways.

The key idea is that the deletion of lonely edges can help parity correction in the following sense.
If a narrow cut $C$ contains exactly two edges of a tree $(V,S)$, then this cut is an $(\odd(S)\symdiff T)$-cut.
Now suppose one of the two edges in $C\cap S$ is a lonely edge of $S$.
Then $|F_S \cap C|=1$ and hence $C$ is not an $(\odd(F_S)\symdiff T)$-cut.
(Note that one can show that it is impossible that both edges in $C\cap S$ are lonely edges of $S$.)

Of course, it might happen that a narrow cut $C$ with $|C\cap S|=2$ does not contain a lonely edge.
However, in this case Seb\H{o}'s analysis of the Best-of-Many-Christofides algorithm is not tight as we will show in Section~\ref{sect:bomc}.
The overall approximation algorithm that we analyze is the following.
Apply Algorithm~\ref{algo:bomc} and Algorithm~\ref{algo:deleting_edges} and return the cheaper of the two resulting $T$-tours.

\section{Analyzing the Best-of-Many Christofides algorithm}\label{sect:bomc}

In this section we present the details of the analysis of the Best-of-Many-Christofides algorithm from~\cite{Seb13}, which builds on~\cite{AnKS15}. 
We include this analysis here for completeness.
Moreover, while analyzing the Best-of-Many-Christofides algorithm, we will introduce some notation and useful facts that we also need later on.

For $S\in \Sscr$ we denote by $I_S$ the unique $T$-join contained in $S$.
Moreover, we define $J_S := S \setminus I_S$. 
Then $J_S$ is the unique $(T\symdiff \odd(S))$-join in $S$.
In the following we write $I_p := \sum_{S\in\Sscr} p_S \cdot \chi^{I_S}$ and $J_p := \sum_{S\in\Sscr} p_S \cdot \chi^{J_S}$.
Then we have \[ x^* = I_p + J_p. \]
The following is well-known.
\begin{lemma}\label{lemma:simple_bound}
 The Best-of-Many-Christofides algorithm (Algorithm~\ref{algo:bomc}) returns a solution to the $T$-tour problem of cost at most 
 \[ 
  c(x^*) + c(J_p).
 \]
\end{lemma}
\begin{proof}
Recall that $J_S$ is an $(T\symdiff \odd(S))$-join in $S$ for every $S\in\Sscr$.
Hence, the cheapest of the $T$-tours $S \cupp J_S^*$ with $p_S > 0$ has cost 
\[ 
\min_{S\in\Sscr: p_S > 0} \left(c(S) + c(J_S^*)\right)\ \le\ \sum_{S\in \Sscr} p_S \left(c(S) + c(J_S^*)\right)\ \le\ \sum_{S\in \Sscr} p_S \left(c(S) + c(J_S)\right)\ =\ c(x^*) + c(J_p).
\]
\end{proof}
If this does not yields an approximation ratio of better than $\frac{8}{5}$, we have $c(J_p) \ge \frac{3}{5} c(x^*)$ and hence
$c(I_p) \le \frac{2}{5} c(x^*)$.
\smallskip 

One important observation by An, Kleinberg, and Shmoys~\cite{AnKS15}
is that for a narrow cut $C$ with $x^*(C)$ much smaller than $2$, a large fraction of the spanning trees will be lonely at $C$.
More precisely, by \eqref{eq:lp_sol_dominates_convex_comb} we have 
\begin{equation}\label{eq:many_lonely_cuts}
 x^*(C) \ge  2 -\sum_{S : C\in \Lscr_S} p_S
\end{equation}
for every narrow cut $C$ and hence 
\begin{equation}\label{eq:few_non_lonely_cuts}
 \sum_{S : C\in \Nscr\setminus \Lscr_S} p_S \le x^*(C) -1.
\end{equation}

Let now $(V,S)$ be a spanning tree. 
Seb\H{o}~\cite{Seb13} uses the following parity correction vector to bound the cost of a cheapest $(\odd(S)\symdiff T)$-join:
\[
 y^S := \frac{1}{2} x^* + \alpha \cdot \chi^{I_S} + \sum_{C\in \Nscr\setminus \Lscr_S} \max\{1-\tfrac{1}{2} x^*(C) - \alpha, 0 \} \cdot v^C, 
\]
where $\alpha \ge 0$ and
\[
 v^C := \frac{1}{2-x^*(C)} \cdot \sum_{S\in \Sscr: C\in \Lscr_S} p_S \cdot \chi^{S\cap C}.
\]
Note that the $T$-join $I_S$ intersects all narrow cuts because all narrow cuts are $T$-cuts.
Moreover, $v^C(C) \ge 1$ for every narrow cut $C$ because of~\eqref{eq:many_lonely_cuts}.

We now show that $y^S$ is indeed a parity correction vector for $S$, i.e.\ it is 
contained in the $(\odd(S) \symdiff T)$-join polyhedron.
\begin{lemma}\label{lemma:parity_correction_bomc}
 For every $(\odd(S) \symdiff T)$-cut $C$ we have $y^S(C) \ge 1$.
\end{lemma}
\begin{proof}
Let $C$ be a cut. 
If $C\notin \Nscr$, we have $y^S(C) \ge \frac{1}{2} x^*(C) \ge 1$.
Otherwise, the constraints of~\eqref{eq:t_tour_lp} imply that $C$ is a $T$-cut.

If $C\in \Lscr_S \subseteq \Nscr$, we have $|S\cap C|=1$ and hence $C$ is also an $\odd(S)$-cut.
Therefore $C$ is not an $(\odd(S) \symdiff T)$-cut.

Now consider the remaining case $C\in \Nscr\setminus \Lscr_S$. 
Since $C$ is a $T$-cut, we have $|I_S \cap C| \ge 1$. 
Using $v^C(C)\ge 1$, this implies
$
 y^S(C) \ge \frac{1}{2} x^*(C) + \alpha + \max\{1-\tfrac{1}{2} x^*(C) - \alpha, 0 \}  \ge 1.
$
\end{proof}
\begin{lemma}\label{lemma:bomc_bound}
Let $\alpha \ge 0$.
Then the Best-of-Many-Christofides algorithm (Algorithm~\ref{algo:bomc}) returns a solution to the $T$-tour problem of cost at most 
 \[ 
  % \frac{3}{2} \cdot c(x^*) + \alpha\cdot c(I_p) + \frac{1}{2}\left(1-\sqrt{2\alpha}\right)^2 \cdot c(L_p).
  \tfrac{3}{2} c(x^*) + \alpha \cdot c(I_p) + \sum_{C\in \Nscr} (x^*(C) -1) \cdot \max\{1-\tfrac{1}{2} x^*(C) - \alpha, 0 \} \cdot c(v^C).
 \]
\end{lemma}
\begin{proof}
By Lemma~\ref{lemma:parity_correction_bomc}, Algorithm~\ref{algo:bomc} returns a $T$-tour of cost at most 
\begin{align*}
 &\min_{S\in \Sscr: p_S > 0} \left( c(S) +  c(y^S) \right) \\
 &\le\ \sum_{S\in\Sscr} p_S \cdot \left( c(S) +  c( y^S) \right) \\
 &\le\ \tfrac{3}{2} c(x^*) + \alpha \cdot c(I_p)
         + \sum_{C\in \Nscr} \sum_{S\in\Sscr : C\notin \Lscr_S}  p_S \cdot\max\{1-\tfrac{1}{2} x^*(C) - \alpha, 0 \} \cdot c(v^C)\\
 &\le\ \tfrac{3}{2} c(x^*) + \alpha \cdot c(I_p) + \sum_{C\in \Nscr} (x^*(C) -1) \cdot \max\{1-\tfrac{1}{2} x^*(C) - \alpha, 0 \} \cdot c(v^C),
\end{align*}
where we used~\eqref{eq:few_non_lonely_cuts} in the last inequality.
\end{proof}

Seb\H{o}~\cite{Seb13} now completes the analysis as follows.
If an edge $e$ of $S$ is lonely at a narrow cut $C$, we have $1 = |C\cap S| \ge |C \cap I_S| \ge 1$
(because $C$ is a $T$-cut and $I_S$ is a $T$-join) and hence $e\in I_S$.
This shows $L_S \subseteq I_S$.
Therefore, with $L_p := \sum_{S\in\Sscr} p_S \cdot \chi^{L_S}$, we have 
\[ \sum_{C\in \Nscr} (2-x^*(C)) \cdot v^C = L_p \le I_p.\]
Using this and $\frac{(x-1)\cdot(1-\frac{1}{2}x-\frac{1}{8})}{2-x}\le \frac{1}{8}$ for $1 \le x < 2$, 
one can show that Lemma~\ref{lemma:bomc_bound} for $\alpha = \frac{1}{8}$ implies an approximation ratio of $\frac{8}{5}$
if $c(I_p) \le \frac{2}{5} \cdot c(x^*)$.
Otherwise, Lemma~\ref{lemma:simple_bound} implies the approximation ratio $\frac{8}{5}$.
This analysis is only tight if $c(L_p) = c(I_p) = \frac{2}{5} c(x^*)$.

\section{Deleting lonely edges for parity correction}\label{sect:deleting_edges}

In this section we show the following.

\begin{lemma}\label{lemma:edge_deletion_bound}
Algorithm~\ref{algo:deleting_edges} returns a solution to the $T$-tour problem of cost at most 
  \[ 
 \frac{8}{5} c(x^*) + \frac{1}{5} c(I_p) - \frac{2}{5} c(L_p)- \frac{2}{5} \sum_{C\in \Nscr} (2-x^*(C))^2 \cdot c(v^C).
 \]
\end{lemma}
Note that the bound in Lemma~\ref{lemma:edge_deletion_bound} is smaller than $\frac{8}{5}c(x^*)$ in the case where Seb\H{o}'s~\cite{Seb13} analysis of the Best-of-Many-Christofides algorithm is tight: then we have $c(L_p) = c(I_p) = \frac{2}{5} c(x^*)$.
\smallskip 

In the rest of this section we prove Lemma~\ref{lemma:edge_deletion_bound}.
Let $(V,S)$ be a spanning tree and $F_S := S \setminus L_S$.
The following lemma bounds the average cost for reconnection. 
It is essentially due to Seb\H{o} and van Zuylen~\cite{SebvZ16} who proved it for the Path TSP.
Their proof can be generalized to the $T$-tour problem as we show below.
\begin{lemma}\label{lemma:reconnection_cost}
 We have $c^S(x^*) - c(x^*) \le 2 \cdot \sum_{C\in \Lscr_S} (x^*(C) -1)\cdot c(S\cap C)$.
\end{lemma}
\begin{proof}
We consider the directed bipartite auxiliary graph with vertex set $E \cupp \Lscr_S$ and arc set 
$A:=\{ (e,C) : e\in E, C \in \Lscr_S, e\in C\}$.
We claim that there exists a function $f: A \to \mathbb{R}_{\ge 0}$ in this auxiliary graph such that $f(\delta^+(e)) \le x^*_e$ for all $e\in E$
and $f(\delta^-(C)) \ge 1$ for all $C\in \Lscr_S$.
By the Hall condition
 such a function $f$ exists if and only if for every subset $\Lscr' \subseteq \Lscr_S$
\begin{equation}\label{eq:hall-condition} 
 x^*\left(\bigcup_{C\in \Lscr'}C\right) \ge |\Lscr'|.
\end{equation}
Let $\Lscr' \subseteq \Lscr_S$ and let $L'\subseteq L_S$ be the set of lonely edges of $S$ that are contained in of the cuts $\Lscr_S$.
Note that every lonely edge $l\in L'$ is contained in only one lonely cut, namely the fundamental cut of $l$ in the tree $S$.
Since every cut in $\Lscr_S$ contains exactly one lonely edge of $S$, this implies $|L'|=|\Lscr'|$.
Let $\Wscr$ be the partition of $V$ that consists of the vertex sets of the connected components of $(V,S \setminus L')$.
Then $|\Wscr| = |\Lscr'| + 1$.  
Moreover, for every cut $C\in \Lscr'$ and every edge $\{v,w\}\in C$, the unique edge $l\in C\cap S$ is contained in the unique $v$-$w$-path in $S$.
Since $l\in L'$, this implies that $v$ and $w$ are contained in different connected components of $(V,S\setminus L')$ and hence $\{v,w\}\in \delta(\Wscr)$.
Using the constraints of the LP~\eqref{eq:t_tour_lp}, we therefore get
\[ 
  x^*\left(\bigcup_{C\in \Lscr'}C\right) \ge  x^*\left(\delta(\Wscr)\right) \ge |\Wscr| -1 = |\Lscr'|.
\]
This shows that a function $f$ as claimed does indeed exist.
Thus we have
\begin{align*}
 c^S(x^*) - c(x^*) =&  \sum_{e\in E} x^*_e \cdot 2 \cdot 
                       \left(\sum_{C\in\Lscr_S : e\in C} c(S\cap C) - \max_{C\in\Lscr_S : e\in C} c(S\cap C) \right) \\
                   \le& \sum_{e\in E} \sum_{C\in\Lscr_S : e\in C} 2x^*_e \cdot c(S\cap C) 
                      - \sum_{e\in E}  2 \cdot \left(\sum_{C\in\Lscr_S : e\in C} f(e,C)\right)
                         \cdot \max_{C\in\Lscr_S : e\in C} c(S\cap C) \\
                   \le& \sum_{e\in E} \sum_{C\in\Lscr_S : e\in C} 2x^*_e \cdot c(S\cap C) 
                       - \sum_{e\in E} 2 \cdot \sum_{C\in\Lscr_S : e\in C} f(e,C) \cdot c(S\cap C) \\
                   =& \sum_{C\in\Lscr_S} 2 \cdot c(S\cap C) \cdot \sum_{e\in C} \left(x^*_e - f(e,C)\right) \\
                   \le& \sum_{C\in\Lscr_S} 2 \cdot c(S\cap C) \cdot \left(x^*(C) - 1\right).
\end{align*}
\end{proof}

We now construct a parity correction vector $\bar y^S$ for $F_S$.
Let
\[
 \bar y^S := \frac{2}{5} x^* + \frac{1}{5} \chi^{S} + \frac{1}{5} \chi^{I_S \setminus L_S} + \sum_{C\in \Lscr_S} \frac{2}{5} (2-x^*(C)) \cdot \chi^{S\cap C}.
\] 
Before formally proving that $\bar y^S$ is a parity correction vector for $F_S$, let us briefly describe the purpose of the different parts.
The term $\frac{2}{5} x^* + \frac{1}{5} \chi^{S}$ is used to ensure $\bar y^S(C)\ge 1$ for every cut that contains at least three edges of $S$ or is not narrow. 
A narrow cut $C$ that contains precisely two edges of $S$ will either have the ``correct parity'' after deleting the lonely edges,
i.e.\ it contains exactly one edge of $F_S$ and is therefore not an $(\odd(F_S)\symdiff T)$-cut, or it contains no lonely edge.
In the latter case, the term $\frac{1}{5} \chi^{I_S \setminus L_S}$ will contribute to $\bar y^S(C)$.
% We remark that Seb\H{o}'s analysis of the Best-of-Many-Christofides algorithm is tight only if $c(I_S\setminus L_S)=0$,
% which we will exploit to bound the cost of this term.
Finally, the last term is used for the lonely cuts of $S$. These are all $(\odd(F_S)\symdiff T)$-cuts.
% Note that we can afford the cost of this last term, because we save even more by deleting the lonely edges.

We now formally prove that $\bar y^S$ is indeed contained in the $(\odd(F_S) \symdiff T)$-join polyhedron.
\begin{lemma}
 For every $(\odd(F_S) \symdiff T)$-cut $C$ we have $\bar y^S(C) \ge 1$.
\end{lemma}
\begin{proof}
 Let $C$ be an $(\odd(F_S) \symdiff T)$-cut.
 If $C$ is not narrow, we have \
 \[
 \bar y^S(C) \ge \tfrac{2}{5} x^*(C) + \tfrac{1}{5} |S \cap C| \ge 
 \tfrac{2}{5} \cdot 2 + \tfrac{1}{5} \cdot 1 = 1.
 \]
 Let now $C\in \Nscr$. 
 Then $C$ is a $T$-cut. 
 Since $C$ is an $(\odd(F_S) \symdiff T)$-cut, this implies $|F_S\cap C|$ even.
 Hence it suffices to consider the following three cases.\\[1mm]
 \emph{Case 1:} $|F_S\cap C|=0$ and $|S\cap C|\le 2$. \\[1mm]
 Because the cut $C$ is narrow, it is a $T$-cut and hence $|I_S \cap C|$ is odd.
 This implies $|I_S \cap C|=1$.
 Recall that $S \setminus F_S = L_S \subseteq I_S$.
 Therefore, $|(S \setminus F_S)\cap C| \le |I_S \cap C|=1$.
 Using $|F_S\cap C|=0$, we conclude $|S\cap C| = 1$ and hence $C\in \Lscr_S$.
 Thus, 
 \[ \bar y^S(C)\ \ge\ \tfrac{2}{5} x^*(C) + \tfrac{1}{5} |S\cap C| + \tfrac{2}{5} (2-x^*(C)) |S\cap C|\ =\ 
 \tfrac{2}{5} x^*(C) + \tfrac{1}{5} + \tfrac{2}{5} (2-x^*(C))\ =\ 1.\]
 \emph{Case 2:} $|S\cap C| = |F_S\cap C|=2$. \\[1mm]
 Because $F_S=S\setminus L_S$, we have $L_S \cap C=\emptyset$.
 Recall that $C\in \Nscr$ and hence $C$ is a $T$-cut, implying $|C\cap I_S| \ge 1$.
 Therefore, $|(I_S \setminus L_S) \cap C| \ge 1$.
 We conclude 
 \[
   \bar y^S(C)\ \ge\ \tfrac{2}{5} x^*(C) + \tfrac{1}{5} |S \cap C| + \tfrac{1}{5} |(I_S \setminus L_S) \cap C|
          \ \ge\  \tfrac{2}{5} \cdot 1 + \tfrac{1}{5} \cdot 2  + \tfrac{1}{5} \cdot 1
          \ =\ 1.
 \]
 \emph{Case 3:} $|S\cap C|\ge 3$. \\[1mm]
 Then $\bar y^S(C) \ge \tfrac{2}{5} x^*(C) + \tfrac{1}{5} |S\cap C| \ge  \tfrac{2}{5} \cdot 1 + \tfrac{1}{5} \cdot 3 = 1$.
\end{proof} 

We now prove Lemma~\ref{lemma:edge_deletion_bound}.
For every edge $e\in S$ we have $c^S(e)=c(e)$.
Therefore, by Lemma~\ref{lemma:reconnection_cost} we have 
\begin{equation}\label{eq:reconnection_cost}
c^S(\bar y^S) - c(\bar y^S)\ \le\ c^S(\tfrac{2}{5} x^*) - c(\tfrac{2}{5} x^*)
                 \ \le\ \sum_{C\in \Lscr_S} \tfrac{4}{5}\cdot (x^*(C)-1) \cdot c(S\cap C).
\end{equation}
Hence,
\begin{align*}
 c^S(\bar y^S) \ \le&\ \tfrac{2}{5} c(x^*) + \tfrac{1}{5} c(S) +  \tfrac{1}{5} c(I_S \setminus L_S) + \sum_{C\in \Lscr_S} 
 \left(\tfrac{2}{5} (2-x^*(C)) +  \tfrac{4}{5} (x^*(C)-1) \right) c(S\cap C)\\
 =&\ \tfrac{2}{5} c(x^*) + \tfrac{1}{5} c(S) +  \tfrac{1}{5} c(I_S \setminus L_S) + \sum_{C\in \Lscr_S} 
 \tfrac{2}{5} x^*(C)\cdot c(S\cap C)\\
 =&\ \tfrac{2}{5} c(x^*) + \tfrac{1}{5} c(S) +  \tfrac{1}{5} c(I_S \setminus L_S) + \tfrac{4}{5} c(L_S)
       -  \tfrac{2}{5} \sum_{C\in \Lscr_S}(2-x^*(C))\cdot c(S\cap C).
\end{align*}
Moreover, $c(F_S)=c(S)-c(L_S)$, implying
\begin{align*}
 c(F_S) +  c^S(\bar y^S)\ \le&\ \tfrac{2}{5} c(x^*) + \tfrac{6}{5} c(S) +  \tfrac{1}{5} c(I_S \setminus L_S) - \tfrac{1}{5} c(L_S) -  \tfrac{2}{5} \sum_{C\in \Lscr_S}(2-x^*(C))\cdot c(S\cap C)\\
 =&\ \tfrac{2}{5} c(x^*) + \tfrac{6}{5} c(S) + \tfrac{1}{5} c(I_S)  - \tfrac{2}{5} c(L_S)-  \tfrac{2}{5} \sum_{C\in \Lscr_S}(2-x^*(C))\cdot c(S\cap C).
\end{align*}
Therefore, the cheapest of the computed $T$-tours has cost at most 
\begin{align*}
 &\hspace*{-5mm}\min_{S\in \Sscr: p_S > 0} \left( c(F_S) +  c^S(\bar y^S) \right) \\[2mm]
 \le&\ \sum_{S\in\Sscr} p_S \cdot \left( c(F_S) +  c^S(\bar y^S) \right) \\
 \le&\ \sum_{S\in\Sscr} p_S \cdot \left( \tfrac{2}{5} c(x^*) + \tfrac{6}{5} c(S) + \tfrac{1}{5} c(I_S)  - \tfrac{2}{5} c(L_S) -  \tfrac{2}{5} \sum_{C\in \Lscr_S}(2-x^*(C))\cdot c(S\cap C)\right) \\
%  =&\ \tfrac{2}{5} c(x^*) + \tfrac{6}{5} c(x^*) + \tfrac{1}{5} c(I_p)  - \tfrac{2}{5} c(L_p) - \tfrac{2}{5} \sum_{C\in \Nscr} (2-x^*(C))^2 \cdot c(v^C)\\
 =&\ \tfrac{8}{5} c(x^*) + \tfrac{1}{5} c(I_p)  - \tfrac{2}{5} c(L_p) - \tfrac{2}{5} \sum_{C\in \Nscr} (2-x^*(C))^2 \cdot c(v^C).
\end{align*}
This completes the proof of Lemma~\ref{lemma:edge_deletion_bound}.

\section{Proof of the overall approximation ratio}\label{sect:overall}

We now combine the results from the previous sections to prove our main result.
\begin{theorem}
There is a polynomial-time algorithm that computes for every instance of the $T$-tour problem
a solution of cost at most $\frac{11}{7}$ times the value of the LP \eqref{eq:t_tour_lp}.
\end{theorem}
\begin{proof}
The algorithm we analyze is the following.
Run Algorithm~\ref{algo:bomc} and Algorithm~\ref{algo:deleting_edges} and return the better of the two resulting $T$-tours.
Then each of Lemma~\ref{lemma:simple_bound}, Lemma~\ref{lemma:bomc_bound}, and Lemma~\ref{lemma:edge_deletion_bound} yields an upper bound on 
the cost of the resulting $T$-tour.
We will take a convex combination of these three upper bounds.
The bounds from Lemma~\ref{lemma:simple_bound}, Lemma~\ref{lemma:bomc_bound}, and Lemma~\ref{lemma:edge_deletion_bound} will be weighted with $\lambda_1, \lambda_2, \lambda_3$, respectively.
We set $\lambda_1 := \frac{2}{21}$, $\lambda_2 := \frac{2}{3}$, $\lambda_3 := \frac{5}{21}$, and $\alpha := \frac{1}{14}$.
%(All of these constants have been optimized numerically.)
Then $\lambda_1, \lambda_2, \lambda_3 \ge 0$ and $\lambda_1 + \lambda_2 + \lambda_3 = 1$.
Hence, our algorithm computes a $T$-tour of cost at most
\begin{align*}
 & \lambda_1 \cdot \Big( c(x^*) + c(J_p) \Big) \\
 & + \lambda_2 \cdot \left(  \tfrac{3}{2} c(x^*) + \alpha \cdot c(I_p) + \sum_{C\in \Nscr} (x^*(C) -1) \cdot \max\{1-\tfrac{1}{2} x^*(C) - \alpha, 0 \} \cdot c(v^C) \right) \\
 & + \lambda_3 \cdot \left( \tfrac{8}{5} c(x^*) + \tfrac{1}{5} c(I_p) - \tfrac{2}{5} c(L_p)- \tfrac{2}{5} \sum_{C\in \Nscr} (2-x^*(C))^2 \cdot c(v^C) \right) 
 \\[2mm]
 =\ & \left(1+ \tfrac{1}{2} \lambda_2 + \tfrac{3}{5} \lambda_3 \right) \cdot c(x^*) + \lambda_1 \cdot c(J_p)  
      + \left( \alpha \cdot \lambda_2 + \tfrac{1}{5} \lambda_3\right) \cdot c(I_p) - \tfrac{2}{5} \lambda_3 \cdot c(L_p) \\
    & + \sum_{C\in \Nscr} \left(  (x^*(C) -1) \cdot \max\{1-\tfrac{1}{2} x^*(C) - \alpha, 0 \} \cdot \lambda_2
                                - \tfrac{2}{5} (2-x^*(C))^2 \cdot \lambda_3 \right) \cdot  c(v^C) 
                                \\[2mm]
 =\ & \tfrac{31}{21} \cdot c(x^*) + \tfrac{2}{21} \cdot c(J_p)  
      + \tfrac{2}{21} \cdot c(I_p) - \tfrac{2}{21} \cdot c(L_p) \\
    & + \sum_{C\in \Nscr} \left(  (x^*(C) -1) \cdot \max\{\tfrac{13}{14} -\tfrac{1}{2} x^*(C), 0 \} \cdot \tfrac{2}{3}
                                - \tfrac{2}{21} \cdot (2-x^*(C))^2 \right) \cdot  c(v^C) 
                                \\[2mm]
 =\ & \tfrac{11}{7}  \cdot c(x^*) - \tfrac{2}{21} \cdot c(L_p) \\
    & + \sum_{C\in \Nscr} \left(  (x^*(C) -1) \cdot \max\{\tfrac{13}{14} -\tfrac{1}{2} x^*(C), 0 \} \cdot \tfrac{2}{3}
                                - \tfrac{2}{21} \cdot (2-x^*(C))^2 \right) \cdot  c(v^C),
\end{align*}
where we used $x^* = J_p + I_p$.
Using $L_p = \sum_{C \in \Nscr} (2-x^*(C)) \cdot v^C$, this implies the following 
upper bound on the cost of our $T$-tour:
\begin{align*}
 &  \tfrac{11}{7} \cdot c(x^*) - \tfrac{2}{21} \cdot c(L_p) 
     + \max_{C\in \Nscr} \left(\tfrac{x^*(C) -1}{2-x^*(C)} \cdot \max\{\tfrac{13}{14} -\tfrac{1}{2} x^*(C),\ 0 \} \cdot \tfrac{2}{3}
                                - \tfrac{2}{21} \cdot (2-x^*(C)) \right) \cdot  c(L_p) 
                                \\[2mm]
  \le\ & \tfrac{11}{7} \cdot c(x^*) - \tfrac{2}{21} \cdot c(L_p) 
     + \max_{1\le x < 2} \left(\tfrac{x -1}{2-x} \cdot \max\{\tfrac{13}{14} -\tfrac{1}{2} x,\ 0 \} \cdot \tfrac{2}{3}
                                - \tfrac{2}{21} \cdot (2-x) \right) \cdot  c(L_p) 
                                \\[2mm]     
   =\ & \tfrac{11}{7} \cdot c(x^*) - \tfrac{2}{21} \cdot c(L_p) 
     + \tfrac{1}{21} \cdot \max_{1\le x < 2} \left(\tfrac{x -1}{2-x} \cdot \max\{13 -7 x,\ 0 \} + 2 x - 4 \right) \cdot  c(L_p) 
                                \\[2mm]                                
  =\ & \tfrac{11}{7} \cdot c(x^*) - \tfrac{2}{21} \cdot c(L_p) + \tfrac{2}{21} \cdot  c(L_p) \\[2mm]
  =\ & \tfrac{11}{7}  \cdot c(x^*),
\end{align*}
where the maximum is attained for $x=\tfrac{5}{3}$.
\end{proof}

\end{document}